\documentclass[twoside,a4paper,11pt]{article}

\usepackage{amssymb}
\usepackage{amsmath}
\usepackage{enumerate}
\usepackage{latexsym}
\usepackage{physics}
\usepackage[numbers,sort&compress]{natbib}

\newtheorem{theorem}{Theorem}
\newtheorem{lemma}[theorem]{Lemma}

 \newtheorem{corollary}[theorem]{Corollary}

 \makeatletter
 \def\section{\@startsection {section}{1}{\z@}{-1.5ex plus -.5ex
 minus -.2ex}{1ex plus .2ex}{\large\bf}}

\newcommand\blfootnote[1]{%
  \begingroup
  \renewcommand\thefootnote{}\footnote{#1}%
  \addtocounter{footnote}{-1}%
  \endgroup
}

\begin{document}
\title{Some constructions of quantum MDS codes}
\author{Simeon Ball }
 \date{}
\maketitle

\begin{abstract}

We construct quantum MDS codes with parameters $ [\![ q^2+1,q^2+3-2d,d ]\!] _q$ for all $d \leqslant q+1$, $d \neq q$. These codes are shown to exist by proving that there are classical generalised Reed-Solomon codes which contain their Hermitian dual. These constructions include many constructions which were previously known but in some cases these codes are new. We go on to prove that if $d\geqslant q+2$ then there is no generalised Reed-Solomon $[n,n-d+1,d]_{q^2}$ code which contains its Hermitian dual. We also construct an $ [\![ 18,0,10 ]\!] _5$ quantum MDS code, an $ [\![ 18,0,10 ]\!] _7$ quantum MDS code and a $ [\![ 14,0,8 ]\!] _5$ quantum MDS code, which are the first quantum MDS codes discovered for which $d \geqslant q+3$, apart from the $ [\![ 10,0,6 ]\!] _3$ quantum MDS code derived from Glynn's code. 
\blfootnote{14 January 2021. The author acknowledges the support of the project MTM2017-82166-P of the Spanish {\em Ministerio de Ciencia e Innovaci\'on.}} 
\end{abstract}

\section{Introduction}

To be able to store information on a quantum system and process that information, it is essential to use a quantum error-correcting code. In this way any small perturbation of the system can be corrected and the information restored. As with classical error-correcting codes, if there is no limit to the size of the alphabet (the local dimension of the quantum particles in the case of a quantum error-correcting code) then we can employ optimal codes which meet the Singleton bound (the quantum Singleton bound for quantum error-correcting codes) which are called maximum distance separable (MDS) codes. Nearly all quantum error-correcting MDS codes known are constructed by means of the more general construction of employing a classical linear code over ${\mathbb F}_{q^2}$ which contains its Hermitian dual. In most cases, this has been done by proving that there are generalised (possibly truncated) Reed-Solomon codes which contain their Hermitian dual. In this article we prove that if the minimum distance $d$ satisfies $d \leqslant q+1$ and $d \neq q$, then there are generalised Reed-Solomon codes of length $q^2+1$ which contain their Hermitian dual. This was not previously known when both $d$ and $q$ are even. More surprisingly, we go on to prove that if $d\geqslant q+2$ then there are no generalised Reed-Solomon codes of any length which contain their Hermitian dual. However, we do provide some sporadic examples of quantum MDS codes for which $d \geqslant q+3$.

\section{Linear and quantum error-correcting codes}

Let $A$ be a finite set. A {\em code} $C$ over $A$ of length $n$ is a subset of $A^n$. The  elements of $C$ are called {\em codewords}. The minimum distance $d$ is the minimum number of coordinates in which two codewords differ. The {\em Singleton bound} states that
$$
|C| \leqslant |A|^{n-d+1}.
$$
A code for which $|C|=|A|^{n-d+1}$ is called a {\em maximum distance separable} code or simply an MDS code.

Let ${\mathbb F}_q$ denote the finite field with $q$ elements, where $q$ is the power of some prime $p$. The {\em weight} of a vector of ${\mathbb F}_q^n$ is the number of non-zero coordinates that it has.

A $k$-dimensional {\em linear code} $C$ of length $n$ is a $k$-dimensional subspace of ${\mathbb F}_q^n$.  We will denote such a code with minimum distance $d$ as an $[n,k,d]_q$ code.

Let $\alpha$ be a Hermitian form defined on ${\mathbb F}_{q^2}^n \times {\mathbb F}_{q^2}^n $ by 
$$
\alpha(u,v)=u_1v_1^{q}+\cdots+u_nv_n^{q}.
$$
The Hermitian dual of a linear code $C$ over ${\mathbb F}_{q^2}$ is
$$
C^{\perp_h}=\{ v \in {\mathbb F}_q^n \ | \ \alpha(u,v)=0,\ \mathrm{for} \ \mathrm{all} \ u \in C \}.
$$

In this article we will construct linear MDS codes $C$ over ${\mathbb F}_{q^2}$ for which $C^{\perp_h} \leqslant C$ and from these we will construct previously unknown quantum MDS codes. 

The following lemma is from \cite[Proposition 2.1]{CMT2018}.

\begin{lemma} \label{hermitedual}
$C$ is a linear $[n,k,n-k+1]_{q^2}$ MDS code if and only if $C^{\perp_h}$ is a linear $[n,n-k,k+1]_{q^2}$ MDS code. Moreover, $(C^{\perp_h})^{\perp_h}= C$.
\end{lemma}

A quantum code on $n$ subsystems is a $K$-dimensional subspace of $({\mathbb C}^q)^{\otimes n}$. A code with minimum distance $d$ is able to detect errors, which act non-trivially on the code space, on up to $d-1$ of the subsystems and correct errors on up to $\frac{1}{2}(d-1)$ of the subsystems. If the dimension $K=q^k$ for some $k$ then we say the quantum code is an $[\![ n,k,d ]\!] _q$ code and if not simply an $ (\!( n,K,d)\!)_q$ code. 

To be able to describe in more detail the quantum error-correcting codes we shall construct here, we have to discuss the type of errors we wish to correct. Let $\{ \ket{x} \ | \ x \in {\mathbb F}_q\}$ be a basis of ${\mathbb C}^q$. We define the following set of endomorphisms of ${\mathbb C}^q$ called the {\em generalised Pauli operators}. For each $a,b\in {\mathbb F}_q$, we define $X(a)$ by its action on the basis vectors, $X(a)\ket{x}=\ket{x+a}$, and likewise $Z(b)$ by $Z(b)\ket{x}=e^{2\pi i \mathrm{tr}(bx)/p} \ket{x}$, where $\mathrm{tr}$ denotes the usual trace map from ${\mathbb F}_q$ to ${\mathbb F}_p$. The generalised Pauli operators are of the form $X(a)Z(b)$, for some $a,b \in {\mathbb F}_q$. In the error model, the (Pauli) errors on $({\mathbb C}^q)^{\otimes n}$ are tensor products of generalised Pauli operators. An error has weight $t$ if precisely $t$ of the components in the tensor product are not the identity operator, whilst the remaining $n-t$ are the identity operator.  A quantum error-correcting code of minimum distance $d$ is able to correct all Pauli errors of weight at most $\frac{1}{2}(d-1)$ which act non-trivially on the code subspace. Such quantum error-correcting codes are most commonly constructed by taking the joint eigenspace of eigenvalue $1$ of a subgroup of Pauli operators. These codes are called {\em stabiliser} codes. It turns out that stabiliser codes are equivalent to certain classical codes which are additive over ${\mathbb F}_{q^2}$. One construction of a particular subset of these codes is that employed in Theorem~\ref{sigmaortog}. Thus the quantum codes which will construct here are stabiliser codes. We refer to \cite{KKKS2006} for a more detailed discussion of stabiliser codes.

We rely on the following theorem from \cite[Corollary 19]{KKKS2006}.

\begin{theorem} \label{sigmaortog}
If there is an $[n,n-k,d]_{q^2}$ linear code $C$ such that $C^{\perp_h} \leqslant C$ then there exists an $ [\![ n,n-2k,d' ]\!] _q$ quantum code with $d' \geqslant d$.
\end{theorem}

The quantum Singleton bound states that for an $ [\![ n,k,d ]\!] _q$ quantum code, $k \leqslant n-2d+2$. A code reaching this bound is called a {\em quantum MDS code}. The quantum Singleton bound implies that for an $ [\![ n,n-2k,d' ]\!] _q$ quantum code, $d' \leqslant k+1$. Thus, Theorem~\ref{sigmaortog} implies for MDS codes that $d'=d$. Hence, we have the following, see \cite[Corollary 3.2]{JLLX2010}.

\begin{theorem} \label{hermort}
If there is an $[n,n-k,k+1]_{q^2}$ linear MDS code $C$ such that $C^{\perp_h} \leqslant C$ then there is an $ [\![ n, n-2k, k+1 ]\!] _q$ quantum MDS code.
\end{theorem}

In the next section we will give a simple construction of a $[q^2+1,k,q^2+2-k]_{q^2}$ MDS codes which is contained in its Hermitian dual, for all $k \leqslant q$, $k \neq q-1$. This code is a generalised Reed-Solomon code. The Hermitian dual of this code is a $[q^2+1,q^2+1-k,k+1]_{q^2}$ generalised Reed-Solomon code which contains its Hermitian dual. Thus, Theorem~\ref{hermort} implies the existence of $ [\![ q^2+1,q^2+1-2k,k+1 ]\!] _q$ quantum MDS codes for all $k \leqslant q$, $k \neq q-1$.  

In Huber and Grassl \cite{HG2019}, we find the following bounds on the existence of an $[\![ n,k,d ]\!] _q$ quantum MDS code for $q=5$. If $n+k=14$ then $4 \leqslant d \leqslant 8$. The lower bound comes from a construction of a $[10,3,8]_{25}$ linear MDS code which is contained in its Hermitian dual. The upper bound would be attained by a $[14,7,8]_{25}$ linear MDS code which is equal to its Hermitian dual. We will construct such a code here. If $n+k=18$ then we have $5 \leqslant d \leqslant 10$. The upper bound would be attained if there is a $[18,9,10]_{25}$ linear MDS code which is equal to its Hermitian dual. Again, we will construct such a code here.

\section{A construction of quantum MDS codes based on generalised Reed-Solomon codes}

There are many constructions of quantum MDS codes with $d \leqslant q+1$, mostly based on cyclic or constacyclic constructions and generalised Reed-Solomon codes. For example those contained in \cite{CLZ2015,FF2018,FF2018b}, \cite{GR2015,HXC2016}, \cite{JKW2017}, \cite{JX2014,JLLX2010,KZ2012,KZL2014}, \cite{LX2010,LXW2008},
\cite{SYC2018,SYW2019,SYZ2017} and \cite{WZ2015,Yan2019,ZC2014,ZG2017}.
 
Here we give a short proof that for all $k \leqslant q$, $k \neq q-1$, there are $ [\![ q^2+1,q^2+1-2k,k+1 ]\!] _q$ quantum MDS codes arising from generalised Reed-Solomon codes of length $q^2+1$. This extends the results of Jin et al \cite{JLLX2010}, who proved the existence for  $k \leqslant \frac{1}{2}q$ and Grassl and R\"otteler \cite{GR2015} who proved the existence for $k \leqslant q$ unless $q$ is even and $k$ is odd. Grassl and R\"otteler \cite{GR2015} also prove the existence for $k=q-1$ and $q\in \{8,16,32,64,128\}$, cases which are not covered by Theorem~\ref{GRSconstr}.
Note that a similar result to Theorem~\ref{GRSconstr} is claimed in \cite[Theorem 4.6]{JX2014} but this relies on the erroneous \cite[Corollary 3.2, part (ii)]{JX2014}. Indeed, the hypothesis of \cite[Theorem 4.6]{JX2014} includes the case $k=q-1$ and $q=4$. The proof involves showing that there is a generalised Reed-Solomon code with parameters $[17,3,15]_{16}$ contained in its Hermitian dual. However, a quick computer search reveals that no such code exists. In fact, computer searches reveal that neither the $[18,3,16]_{16}$ coming from a regular or a Lunelli-Sce hyperoval of PG$(2,16)$ (where PG$(k-1,q)$ denotes the $(k-1)$-dimensional projective space over ${\mathbb F}_q$) can be truncated to a $[17,3,15]_{16}$ code which is contained in its Hermitian dual. In other words, the puncture code has no codewords of weight $17$, or in fact any odd weight. We refer to \cite{GR2015} for more on constructions of quantum MDS codes via the puncture code. The only other complete arc of PG$(2,16)$ of size at least $13$ is the Kestenband arc (from \cite{Kestenband1981}) constructed as the intersection of two Hermitian curves. Again, by employing a computer search, we discover that any $[13,3,11]_{16}$ code obtained from this Kestenband arc is not contained in its Hermitian dual. These are the only complete arcs in PG$(2,16)$ of size at least 13 (see for example \cite{HS1998}), so all linear $[n,3,n-2]_{16}$ MDS codes with $n \geqslant 13$ come from one of these examples. We conclude that there is no way to construct an $ [\![ n,n-6,4]\!] _4$ quantum MDS code, for $n \in \{13,15,17\}$, from a linear $[n,3,n-2]_{16}$ MDS code contained in its Hermitian dual. To clarify, any such code must be linearly equivalent to (the truncation of) either a Reed Solomon code, a Kestenband arc or the Lunelli-Sce hyperoval and none of these examples can be truncated to give such a code.

\begin{theorem} \label{GRSconstr}
For $q$ a prime power, there exists a $ [\![ q^2+1,q^2+1-2k,k+1 ]\!] _q$ quantum MDS code for all $k \leqslant q$ where $k \neq q-1$.
\end{theorem}

\begin{proof}
For all $k \leqslant q$ ,where $k \neq q-1$, we will construct a $[q^2+1,k,q^2+2-k]_{q^2}$ MDS code $D$ such that $D \leqslant D^{\perp_h}$. Lemma~\ref{hermitedual} implies that $C=D^{\perp_h}$ is a $[q^2+1,q^2+1-k,k+1]_{q^2}$ linear MDS code such that $C^{\perp_h} \leqslant C$. Theorem~\ref{hermort} then implies that there exists a $ [\![ q^2+1,q^2+1-2k,k+1 ]\!] _q$ quantum MDS code.

Denote by $\{a_1,\ldots,a_{q^2}\}$ the elements of ${\mathbb F}_{q^2}$.

Let $h(X)$ be a monic polynomial of ${\mathbb F}_{q^2}[X]$ of degree $q-k$ such that $h(a_j) \neq 0$ for all $j$. If $k=q$ then we can take $h(X)=1$. If $k \leqslant q-2$ then we can take $h(X)$ to be an  irreducible polynomial in ${\mathbb F}_{q^2}[X]$. 


Define
$$
D=\{(h(a_1)f(a_1),\ldots,h(a_{q^2})f(a_{q^2}),f_{k-1}) \ | \ f \in {\mathbb F}_{q^2}[X],\ \deg f \leqslant k-1\},
$$
where $f_i$ denotes the coefficient of $X^i$ in $f(X)$.

Firstly, we prove that $D$ is an MDS code. Observe that $D$ is a $[q^2+1,k,d]_{q^2}$ linear code, so we have to prove that $d=q^2+2-k$.

Consider $u$ and $v$ two codewords of $D$ given respectively by polynomials $f$ and $g$ of degree at most $k-1$.

If $f_{k-1} \neq g_{k-1}$ and $u$ and $v$ agree in the coordinate indexed by $a_i$ then $a_i$ is a zero of $h(X)(f(X)-g(X))$. Since $(f-g)(X)$ has at most $k-1$ zeros and $h(a_i)\neq0$, the codewords $u$ and $v$ agree in at most $k-1$ coordinates.

If $f_{k-1}= g_{k-1}$ then $(f-g)(X)$ has degree at most $k-2$ and therefore has at most $k-2$ zeros. Thus, the codewords $u$ and $v$ agree in at most $k-1$ coordinates.

Hence, the minimum distance of $D$ is at least $q^2+1-(k-1)$ which attains the Singleton bound. 

Let $r_m$ denote the coefficient of $X^m$ in $h(X)^{q+1}$. Then
$$
\alpha(u,v)=f_{k-1}g_{k-1}^q+\sum_{t \in {\mathbb F}_{q^2}} h(t)^{q+1}f(t)g(t)^q=f_{k-1}g_{k-1}^q+ \sum_{t \in {\mathbb F}_{q^2}} \sum_{m=0}^{(q-k)(q+1)} \sum_{i,j=0}^{k-1} f_ig_j^q r_m  t^{i+jq+m}
$$
$$
=f_{k-1}g_{k-1}^q+\sum_{m=0}^{(q-k)(q+1)} \sum_{i,j=0}^{k-1}f_ig_j^q r_m  \sum_{t \in {\mathbb F}_{q^2}}t^{i+m+jq}
$$
Since $\sum_{t \in {\mathbb F}_{q^2}} t^i=0$ for all $i=0,\ldots,q^2-2$ and $\sum_{t \in {\mathbb F}_{q^2}} t^{q^2-1}=-1$, we have that $\alpha(u,v)=0$. Note that $r_{(q-k)(q+1)}=1$, since $h$ is monic.

Thus, $D \leqslant D^{\perp_h}$, as required.
\end{proof}

The codes constructed in Theorem~\ref{GRSconstr} are examples of generalised Reed-Solomon codes. A generalised Reed-Solomon code over ${\mathbb F}_{q^2}$ is a code
$$
D=\{(c_1f(a_1),\ldots,c_{q^2}f(a_{q^2}),c_{q^2+1}f_{k-1}) \ | \ f \in {\mathbb F}_{q^2}[X],\ \deg f \leqslant k-1\},
$$
where $c_1,\ldots,c_{q^2+1}$ are elements of ${\mathbb F}_{q^2}$. In Theorem~\ref{GRSconstr}, we proved that there are generalised Reed-Solomon codes over ${\mathbb F}_{q^2}$ which are contained in their Hermitian duals in which $c_i=h(a_i)$ for some polynomial $h(X)$. However, for larger dimensions generalised Reed-Solomon codes are not contained in their Hermitian duals, as we shall now prove.

\begin{theorem} \label{noGRS}
If $k \geqslant q+1$ then a $k$-dimensional generalised Reed-Solomon code over ${\mathbb F}_{q^2}$ is not contained in its Hermitian dual.
\end{theorem}

\begin{proof}
Let $A=\{a_1,\ldots,a_{q^2}\}$ be the set of elements of ${\mathbb F}_{q^2}$, where $a_{q^2}=0$. Let $D$ be a generalised Reed-Solomon code over ${\mathbb F}_{q^2}$ which is contained in its Hermitian dual.

Since $D$ is a generalised Reed-Solomon code, there are $v_1,\ldots,v_n$, elements of ${\mathbb F}_{q^2}$, such that
$$
D=\{(v_1f(a_1),\ldots,v_{q^2}f(a_{q^2}),c_f) \ | \ f \in {\mathbb F}_{q^2}[X],\ \deg f \leqslant k-1\},
$$
where $c_f$ is the coefficient of $X^{k-1}$ of $f$ or possibly zero. We also allow some of the $v_i$ to be zero too, which would be equivalent to taking a shorter length generalised Reed-Solomon code.

Consider $u$ and $v$ two codewords of $D$ given respectively by polynomials $f$ and $g$ of degree at most $k-1$.

Since $D$ is contained in its Hermitian dual,
$$
c_fc_g^q+\sum_{i =1}^{q^2} v_i^{q+1} f(a_i)g(a_i)^q=0.
$$
Let $f_i$ denote the coefficient of $X^i$ in $f(X)$ and $g_m$ denote the coefficient of $X^m$ in $g(X)$. Then the above is
$$
c_fc_g^q+\sum_{i =1}^{q^2} \sum_{j=0}^{k-1} \sum_{m=0}^{k-1}v_i^{q+1} f_j g_m^q a_i^{mq+j}=0.
$$

For all $\ell=1,\ldots,q^2-2$, where $\ell=\ell_0+\ell_1q$, with $f(X)=X^{\ell_0}$ and $g(X)=X^{\ell_1}$, this implies that
$$
\sum_{i =1}^{q^2-1} v_i^{q+1} a_i^{\ell}=0.
$$
Considering this set of equations in matrix form $\mathrm{A}u$=0, where $\mathrm{A}=(b_{\ell i})$ is a $(q^2-2) \times (q^2-1)$ matrix given by $b_{\ell i}=a_i^{\ell}$ and where $u$ is the vector of ${\mathbb F}_{q}^{q^2-1}$ whose $i$-th coordinate  is $v_i^{q+1}$. 

The matrix $\mathrm{A}$ contains $(q^2-2) \times (q^2-2)$ submatrix which is a Vandermonde matrix, so has rank $q^2-2$. Therefore, the solution space has dimension one and is spanned by the all-one vector. This implies that $v_i^{q+1}=\lambda^{q+1} \neq 0$ for all $i=1,\ldots,{q^2-1}$, for some $\lambda \in {\mathbb F}_{q^2}$. 

Thus, we have that
\begin{equation} \label{hermsum}
c_fc_g^q+(v_{q^2}^{q+1}-\lambda^{q+1}) f_0 g_0^q-\lambda^{q+1}  \sum_{j,m : q^2-1|j+mq \neq 0}f_j g_m^q =0.
\end{equation}
Since $c_f=f_{k-1}$ or zero, with $f(X)=g(X)=1$, (\ref{hermsum}) implies $v_{q^2}^{q+1}=\lambda^{q+1}$. Then, with $f(X)=g(X)=X^{q-1}$, (\ref{hermsum}) implies $\lambda=0$, a contradiction.

\end{proof}

Since the dual of a generalised Reed-Solomon code is a generalised Reed-Solomon code \cite[Chapter 10]{MS1977}, Theorem~\ref{noGRS} tells us that if $k \geqslant q+1$ then there are no $[n,n-k,k+1]_{q^2}$ generalised Reed-Solomon codes which contain their Hermitian dual. Thus, we should look elsewhere if we want to construct quantum MDS codes via Theorem~\ref{hermort}, for $d \geqslant q+2$. This we will do in the next section.

\section{Linear codes of rate one half}

Let $\mathrm{M}$ be a $k \times k$ matrix with entries from ${\mathbb F}_q$ and let $\mathrm{I}_k$ denote the $k \times k$ identity matrix. For some non-zero $\lambda \in {\mathbb F}_q$, let
$$
\mathrm{G}
=(\lambda \mathrm{I}_k \ | \ \mathrm{M}).
$$
The subspace spanned by the rows of $\mathrm{G}$ is a $k$-dimensional linear code $C$ of length $2k$ over ${\mathbb F}_q$. In the following theorems we will determine the minimum distance of $C$ depending on certain hypotheses regarding $\mathrm{M}$. 

\begin{theorem} \label{nohyp}
If every $j \times (k-d+1+j)$ sub-matrix of $\mathrm{M}$ for $j=1,\ldots,d-1$ has rank $j$ then $C$ is a $[2k,k,\geqslant d]$ code.
\end{theorem} 

\begin{proof}
If $C$ does not have minimum distance at least $d$ then there are two distinct codewords $u$ and $v$ which differ in at most $d-1$ coordinates. Therefore, $u-v=(a_1,\ldots a_k)\mathrm{G}$ has at least $2k-d+1$ zeros. 

Let $S$ be the set of of columns of  $\mathrm{G}$ viewed as points of $\mathrm{PG}(k-1,q)$, the $(k-1)$-dimensional projective space over ${\mathbb F}_q$. Since $(a_1,\ldots , a_k)\mathrm{G}$ has at least $2k-d+1$ zeros, there is a subset $S'$ of $S$ of $2k-d+1$ points which are incident with the hyperplane $a_1X_1+\cdots+a_kX_k=0$.

For any subset $T$ of $S$, let $\mathrm{G}(T)$ denote the submatrix of $\mathrm{G}$ restricted to the columns of $T$. Observe that the rank of the matrix $\mathrm{G}(S')$ is at most $k-1$.

Suppose that $i$ of the points of $S'$ are in the canonical basis and let $S''$ be a subset of $S'$ consisting of the $2k-d+1-i$ points not in the canonical basis. Since $2k-d+1-i \leqslant k$, we have $i \geqslant k-d+1$. Since $\mathrm{G}(S')$ has rank at most $k-1$, the matrix $\mathrm{G}(S'')$ contains a sub-matrix of $\mathrm{M}$ which is a $(k-i) \times (2k-d+1-i)$ sub-matrix of rank at most $k-i-1$. The hypothesis then implies $k-i \geqslant d$ which implies $i \leqslant k-d$, a contradiction.
\end{proof}

\begin{theorem} \label{inversehyp}
Suppose that $\mathrm{M}$ is a non-singular $k \times k$ matrix. If every $j \times (k-d+1+j)$ sub-matrix of $\mathrm{M}$  and $\mathrm{M}^{-1}$ has rank $j$ for all $j=1,\ldots,\lfloor\frac{1}{2}(d-1)\rfloor$ then $C$ is a $[2k,k,\geqslant d]$ code.
\end{theorem} 

\begin{proof}
In the proof of Theorem~\ref{nohyp}, we can assume that $i \geqslant k- \lfloor \frac{1}{2}(d-1) \rfloor$ by multiplying the matrix $\mathrm{G}$ by $\mathrm{M}^{-1}$. Multiplying by $\mathrm{M}^{-1}$ constitutes a change of basis but does not affect the geometry of the point set $S$ and in particular its intersection with hyperplanes. The hypothesis now implies $k-i > \lfloor\frac{1}{2}(d-1)\rfloor$, a contradiction.

\end{proof}

Let $\sigma$ be an automorphism of ${\mathbb F}_q$. For a matrix $\mathrm{A}=(a_{ij})$, we define $\mathrm{A}^{\sigma}=(a_{ij}^{\sigma})$ and $\mathrm{A}^{t}=(a_{ji})$

\begin{theorem} \label{inverseissigma}
Suppose that $\mathrm{M}$ is a non-singular $k \times k$ matrix and that $(M^{\sigma})^t=-\lambda^{\sigma+1} \mathrm{M}^{-1}$ for some automorphism $\sigma$ of ${\mathbb F}_q$ and for some non-zero $\lambda \in {\mathbb F}_{q}$. If every $j \times (k-d+1+j)$ sub-matrix of $\mathrm{M}$ and every $(k-d+1+j)  \times j$ sub-matrix of $\mathrm{M}$ has rank $j$ for all $j=1,\ldots,\lfloor\frac{1}{2}(d-1)\rfloor$ then $C$ is a $[2k,k,\geqslant d]$ code such that $C \leqslant C^{\perp_{\sigma}}$, where $\perp_{\sigma}$ is defined with respect to the sesqui-linear form
$$
\beta(u,v)=u_1v_1^{\sigma}+\cdots+u_{2k}v_{2k}^{\sigma}.
$$
\end{theorem} 

\begin{proof}
Observe that for a sub-matrix $B$ of $\mathrm{A}$ there is a corresponding submatrix $\mathrm{B}'$ of $\mathrm{A}^{\sigma}$ (taking the same restriction to rows and columns) and that the rank of $\mathrm{B'}$ is equal to the rank of $\mathrm{B}$. That $C$ is a linear code of minimum distance $d$ now follows from Theorem~\ref{inversehyp}.

Suppose that $ \mathrm{M}^{-1}=(b_{ij})$.

If $u$ is the $i$-th row of $\mathrm{G}$ and $v$ is the $j$-th row of $\mathrm{G}$, $i \neq j$, then
$$
\beta(u,v)=\sum_{m=1}^k a_{im}a_{jm}^{\sigma}=-\lambda^{\sigma}\lambda \sum_{m=1}^k a_{im} b_{mj}=0.
$$
Meanwhile,
$$
\beta(u,u)=\lambda^{\sigma}\lambda+\sum_{m=1}^k a_{im}a_{im}^{\sigma}=\lambda^{\sigma}\lambda-\lambda^{\sigma}\lambda\sum_{m=1}^k a_{im}b_{mi}=0.
$$
Thus, $C \leqslant C^{\perp_{\sigma}}$.

\end{proof}

We are now in a position to prove the main theorem of this section.  Observe that we now replace $q$ by $q^2$ and the automorphism $\sigma$ will be $x \mapsto x^q$.

\begin{theorem} \label{2kquantum}
If there is a $k \times k$ non-singular matrix $\mathrm{M}$ with entries from ${\mathbb F}_{q^2}$ for which $(M^{\sigma})^t=-\lambda^{q+1} \mathrm{M}^{-1}$ for some non-zero $\lambda \in {\mathbb F}_{q^2}$, where $\sigma(x)=x^q$, and every $j \times j$ sub-matrix of $\mathrm{M}$ is non-singular for all $j=1,\ldots,\lfloor\frac{1}{2}k\rfloor$ then there is a $ [\![ 2k-r,r,k+1-r ]\!] _q$ quantum MDS code, for all $r=0,\ldots,k-1$.
\end{theorem}

\begin{proof}
By Theorem~\ref{inverseissigma}, $C\leqslant C^{\perp_h}$ and $C$ is a $[2k,k,k+1]_{q^2}$ linear MDS code. Since $\dim C=\dim C^{\perp_h}$, we have that $C= C^{\perp_h}$, so Theorem~\ref{hermort} implies there exists a $ [\![ 2k,0,k+1 ]\!] _q$ quantum MDS code.

Rains \cite{Rains1998} showed that if there is a pure $ [\![ n,k,d ]\!] _q$ quantum code with $n,d \geqslant 2$ then there exists a pure $ [\![ n-1,k+1,d-1 ]\!] _q$ quantum code. Later, Rains \cite{Rains1999} proved that a quantum MDS code must be pure, so there exists a $ [\![ 2k-r,r,k+1-r ]\!] _q$ quantum MDS codes, if the hypothesis on $\mathrm{M}$ is satisfied, for all $r=0,\ldots,k-1$.
\end{proof}

\section{Circulant matrices}

A circulant matrix $\mathrm{M}$ is a matrix whose $(i+1)$-st row is a cyclic shift of the first row $i$ places to the right with wrap around. In other words,
$$
\mathrm{M}=\left(\begin{array}{cccccc}
x_1 & x_2 & \ldots & x_{k-1} & x_k \\
x_k & x_1 & \ldots & x_{k-2} & x_{k-1} \\
. & . &  \ddots & .  & . \\
. & . & . &   \ddots  & . \\
x_{2} & x_3 & \ldots & x_{k} & x_1 \\
\end{array} \right)
$$
for some $x=(x_1,\ldots,x_k)$.

A linear code of rate one half which is the row span over ${\mathbb F}_q$ of a matrix 
$$
\mathrm{G}
=(\lambda \mathrm{I}_k \ | \ \mathrm{M}).
$$
for some non-zero $\lambda \in {\mathbb F}_q$, is called a {\em doubly circulant code}. Such codes have been well studied, see for example \cite{Beenker1980} and \cite{VR1985}.

\begin{theorem} \label{sigmacirculant}
Let $\mathrm{M}$ be a $k \times k$ non-singular circulant matrix with entries from ${\mathbb F}_{q^2}$ whose first row is $x=(x_1,\ldots,x_k)$. Then $(\mathrm{M}^{\sigma})^t=-\lambda^{q+1}\mathrm{M}^{-1}$ for some non-zero $\lambda \in {\mathbb F}_{q^2}$, where $\sigma(x)=x^q$ if and only if $H_m(x)=0$  for all $m =1,\ldots,\lfloor \frac{1}{2} k \rfloor$, where
$$
H_m(x)=\sum_{i=1}^k x_i x_{i+m}^q
$$ 
and the indices are read modulo $k$.
\end{theorem}

\begin{proof}
If $\mathrm{M}=(a_{ij})$ then $a_{ij}=x_{j-i+1}$, since $\mathrm{M}$ is circulant whose first row is $x=(x_1,\ldots,x_k)$.  

The scalar product of the $i$-th row of $\mathrm{M}$ with the $r$-th column of $(\mathrm{M}^{\sigma})^t$ is
$$
\sum_{j=1}^k a_{ij} a_{rj}^q=\sum_{j=1}^k x_{j-i+1}x_{j-r+1}^{q}=\sum_{j=1}^k x_j x_{j+i-r}^q=H_{i-r}(x).
$$
Observe that $H_{0}(x) \in {\mathbb F}_q$, so there is a $\lambda \in {\mathbb F}_{q^2}$ such that $\lambda^{q+1}= H_{0}(x)$. Moreover, since $\mathrm{M}$ is non-singular, $\lambda \neq 0$.
$$
H_m(x)^q=\sum_{i=1}^k x_i^q x_{i+m}=\sum_{i=1}^k x_{i-m}^q x_{i}=H_{-m},
$$
so it suffices that $H_m(x)=0$  for $m =1,\ldots,\lfloor \frac{1}{2} k \rfloor$.
\end{proof}

\begin{corollary} \label{maincor}
Let $\mathrm{M}$ be a $k \times k$ non-singular circulant matrix with entries from ${\mathbb F}_{q^2}$ whose first row is $x=(x_1,\ldots,x_k)$. If, for all $m =1,\ldots,\lfloor \frac{1}{2} k \rfloor$, all $m\times m$ submatrices of $\mathrm{M}$ are non-singular,
$$
H_m(x)=\sum_{i=1}^k x_i x_{i+m}^q=0
$$
and
$$
\sum_{i=1}^k x_{i}^{q+1} \neq 0
$$
then there exists a $ [\![ 2k-r,r,k+1-r ]\!] _q$ quantum MDS code, for all $r=0,\ldots,k-1$.
\end{corollary}

\begin{proof}
This follows from Theorem~\ref{2kquantum} and Theorem~\ref{sigmacirculant}.
\end{proof}

\section{Computational results}

Recall that we are interested in constructing $ [\![ 2k,0,k+1 ]\!] _q$ quantum MDS codes with $k \geqslant q+1$. This then implies there are $ [\![ 2k-r,r,k+1-r ]\!] _q$ quantum MDS codes for all $r=0,\ldots,k-1$.


\underline{$k=5$ ($q \leqslant 4$).} 

A $ [\![ 10,0,6 ]\!] _q$ code exists for both $q=3$ and $q=4$, since in both cases there are $[10,5,6]_q$ codes which are equal to their Hermitian dual. These examples are due to Glynn \cite{Glynn1986} for $q=3$ and Grassl and Gulliver \cite{GG2008} for $q=4$.

For $q=3$, Corollary~\ref{maincor} applies if $x=(\epsilon^2, \epsilon^3, \epsilon^3, \epsilon^2, 1)$, where $\epsilon^2=\epsilon+1$. 

For $q=4$, Corollary~\ref{maincor} applies if $x=(\epsilon^2, \epsilon^{12}, \epsilon^{12}, \epsilon^2, 1)$, where $\epsilon^4=\epsilon+1$. 

\underline{$k=6$ ($q \leqslant 5$).} 

An exhaustive search reveals that there are no circulant matrices $\mathrm{M}$ for which the hypothesis of Corollary~\ref{maincor} is satisfied for $q=4$ or $q=5$.

There are examples for $q=7$, for example $x=(\epsilon^{21}, \epsilon^{44},  \epsilon^8, \epsilon^9, \epsilon^{12}, 1)$, where $\epsilon^2=\epsilon+4$, which are perhaps interesting in that they are not obtained by truncating a generalised Reed-Solomon code. This can be checked by calculating the dimension of the subspace of quadrics which are zero on the columns of $\mathrm{G}$. The dimension of the subspace of quadrics which are zero on the columns of a generator matrix of a generalised Reed-Solomon code is ${k-1 \choose 2}$, see Glynn \cite[Theorem 3.3]{Glynn1994}. If the $k$-dimensional MDS code has length at least $2k+1$ then the converse of this statement is also true, see \cite[Theorem 6]{BP2020}. In all the examples in which $\mathrm{M}$ satisfies the hypothesis of Corollary~\ref{maincor}, the dimension is $9$. The existence of a $ [\![ 12,0,7 ]\!] _7$ quantum MDS code was already known, see \cite{GR2015}.

\underline{$k=7$ ($q \leqslant 5$).} 

An exhaustive search reveals that there are no matrices $\mathrm{M}$ for which the hypothesis of Corollary~\ref{maincor} is satisfied for $q=4$.

For $q=5$, Corollary~\ref{maincor} applies if $x=(\epsilon^{10}, \epsilon^{10},1, \epsilon^6, \epsilon^3, \epsilon^6, 1)$, where $\epsilon^2=\epsilon+3$. 

Thus, there is a $ [\![ 14,0,8 ]\!] _5$ quantum MDS code. This was not previously known.

There are examples for $q=7$, for example $x=(\epsilon^{4}, \epsilon^{40},  \epsilon^{45}, 1,1,\epsilon^{45}, \epsilon^{40})$, where $\epsilon^2=\epsilon+4$. Again, as in the case $k=6$, these are perhaps interesting because they cannot be obtained from truncating a generalised Reed-Solomon code. The existence of a $ [\![ 14,0,8 ]\!] _7$ quantum MDS code was already known, see \cite{GR2015}.

\underline{$k=8$ ($q \leqslant 7$).} 

An exhaustive search was too large to perform and no examples were found under the assumption $x_7=x_6$.

An exhaustive search was too large to perform for $k \geqslant 9$. However, under the assumption $x_j=x_{k+2-j}$ further exhaustive searches were executed. This is a natural assumption to make since it is equivalent to assuming $H_m=H_{-m}$. In other words, we assume that $H_m$ is a Hermitian surface, for all $m$. This is equivalent to assuming that $\mathrm{M}$ is symmetric.

Observe that under the assumption that $\mathrm{M}$ is symmetric
 we are obliged to restrict our attention to $k$ odd, since $\mathrm{M}=(a_{ij})$ contains a submatrix
 $$
 \left(\begin{array}{cc} a_{1r} & a_{1,k+2-r} \\ a_{\frac{1}{2}k+1,r} & a_{\frac{1}{2}k+1,k+2-r} \end{array} \right)
=\left(\begin{array}{cc} x_{r} & x_r \\ x_{r+\frac{1}{2}k} & x_{r+\frac{1}{2}k} \end{array} \right),
 $$
 which has zero determinant.

\underline{$k=9$ ($q \leqslant 8$).} 

An exhaustive search for symmetric matrices satisfying the hypothesis of Corollary~\ref{maincor} reveals that there are examples for $q=5$ and $q=7$ and none for $q=8$.

For $q=5$, we have $x=(1, \epsilon^{14}, \epsilon^{21},\epsilon^{16}, \epsilon^{17},\epsilon^{17}, \epsilon^{16}, \epsilon^{21}, \epsilon^{14})$, where $\epsilon^2=\epsilon+3$. 


For $q=7$, we have $x=(1,  \epsilon^{12}, \epsilon^2, \epsilon^{17}, \epsilon^{13}, \epsilon^{13}, \epsilon^{17}, 
  \epsilon^2, \epsilon^{12})$, where $\epsilon^2=\epsilon+4$.
  
Thus, there is an $ [\![ 18,0,10 ]\!] _5$ and an $ [\![ 18,0,10 ]\!] _7$ quantum MDS code. These were not previously known.

\section{Acknowledgments}

The author thanks Felix Huber for many fruitful discussions about quantum codes and Markus Grassl for some helpful comments. The comments and suggestions made by the referees were very much appreciated.

\end{document}